\documentclass[12pt]{article}

\usepackage{epsfig}

\usepackage{amssymb}
\def\be{\begin{equation}}
\def\ee{\end{equation}}
\def\ba{\begin{array}{c}}
\def\ea{\end{array}}
\def\p{\partial}
\def\ben{$$}
\def\een{$$}

\newcommand{\bea}{\begin{eqnarray}}
\newcommand{\eea}{\end{eqnarray}}

\newcommand{\kt}{\rangle}
\newcommand{\br}{\langle}

\newtheorem{thm}{Theorem}

\newtheorem{lemma}[thm]{Lemma}
\newtheorem{prop}[thm]{Proposition}
\newtheorem{defn}[thm]{Definition}

\newenvironment{proof}{\noindent {\bf Proof.}}{\hfill$\square$\vspace{3mm}\endtrivlist}

\begin{document}

\begin{center}

{\Large \bf

Non-Hermitian Heisenberg representation

  }

\vspace{9mm}

{Miloslav Znojil}

\vspace{9mm}

Nuclear Physics Institute CAS, 250 68 \v{R}e\v{z}, Czech Republic

{znojil@ujf.cas.cz}

{http://gemma.ujf.cas.cz/\~{}znojil/}

\end{center}

\section*{Abstract}

In a way paralleling the recently accepted non-Hermitian version of
quantum mechanics in its Schr\"{o}dinger representation (working
often with the innovative and heuristically productive concept of
${\cal PT}-$symmetry), it is demonstrated that it is also possible
to construct an analogous non-Hermitian version of quantum mechanics
in its Heisenberg representation.

\section*{Keywords}

quantum mechanics; non-Hermitian representation of observables;
generalized Heisenberg equations;

\newpage

\section{Introduction \label{oh} }

One of the paradoxes connected with the birth of quantum theory
\cite{nine} is that it was first formulated in its less economical,
matrix version called Heisenberg representation (HREP,
\cite{Heisen}). Naturally, almost immediately there emerged its
amended, equivalent but more economical vectorial version carrying
the name of Schr\"{o}dinger representation (SREP, \cite{Schroe}).

An explanation of the apparent paradox may be found in the much more
friendly correspondence of the quantum HREP picture with its
classical limit. As a consequence, an additional feature of the
paradox is that the HREP formulation of quantum theory survived in
the competition. During cca ninety years of its existence the HREP
language found important specific and unique applications,
typically, in the models using creation and annihilation operators
\cite{Messiah} and, in particular, in the context of relativistic
quantum field theory \cite{Drell}.

We feel inspired by a certain loss of ambition of the HREP
theoreticians and users in the newly emergent field of the study of
quantum observables in their so called non-Hermitian representations
\cite{book}. This loss of ambition resulted in a loss of the
historical leadership because the innovated quantum theory seems
currently available just in its SREP non-Hermitian version (SRNH,
see, e.g., reviews \cite{Carl,ali}). Thus, we believe that it is
time to describe and discuss also a parallel upgrade of the theory
in its alternative, non-Hermitian HREP-based (NHH) version.

The first steps in this direction are to be made in what follows.

\section{Quantum theory using time-independent Hilbert spaces}

From the historical point of view the SREP-based transition to
non-Hermitian Hamiltonians and/or other observables was introduced,
almost 60 years ago, by Freeman Dyson \cite{Dyson}. Unfortunately,
it took several decades before the idea found another and,
incidentally, most successful series of applications in nuclear
physics (cf. review plus basic SRNH theory in Ref.~\cite{Geyer}).
The reason of delay lied, obviously, in a perceivable increase of
mathematical complexity during the SREP $\to$ SRNH change of the
perspective. In contrast, the main source of the persuasive delayed
success of the innovated non-Hermitian theory was truly unexpected.
After the SREP $\to$ SRNH upgrade of the algorithms people
reported a perceivable {\em simplification} of numerical
computations~\cite{Geyer}.

The latter observation returned attention to various unitary
time-evolution SREP prescriptions with difficult (i.e., e.g.,
manybody) Hamiltonians
$\mathfrak{h}_{(SREP)}=\mathfrak{h}_{(SREP)}^\dagger\,$,
 \be
 {\rm i}\p_t|\varphi^{(SREP)}(t)\kt =
 \mathfrak{h}_{(SREP)}\,|\varphi^{(SREP)}(t)\kt\,,
 \  \ \ \ \ \
 %\Omega|\psi(t)\kt=
 |\varphi^{(SREP)}(t)\kt
 \in {\cal H}^{(P)}\,.
 \label{uneSEtaune}
 \ee
According to conventional textbooks the Hamiltonians may be
complicated but they {\em must} be self-adjoint in the preselected
Hilbert space of states in which the superscript $^{(P)}$ stands for
``physical''. In the subsequent step the Dyson's idea has been
recalled. Its essence lies in an inclusion of a part of our {\it a
priori} knowledge of the structure of $|\varphi^{(SREP)}(t)\kt$ into
a {\em non-unitarily} pre-conditioned ansatz
 \be
 |\varphi^{(SREP)}(t)\kt =\Omega|\psi^{(SRNH)}(t)\kt\,,
 \ \ \ \ \ \
 |\psi^{(SRNH)}(t)\kt \in {\cal H}^{(F)}
 \,
 \label{2}
 \ee
with a nontrivial operator $ \Omega^\dagger\Omega=\Theta\neq I$
called metric.

Under some reasonable assumptions about the linear and
time-independent operators $\Omega$ and $\Theta$ the latter ansatz
converts Eq.~(\ref{uneSEtaune}) into its formally equivalent SRNH
alternative \cite{Geyer}
 \be
 {\rm i}\partial_t\,|\psi^{(SRNH)}(t)\kt
 = H_{(SRNH)}\,|\psi^{(SRNH)}(t)\kt \,,
 \ \
 |\psi^{(SRNH)}(t)\kt \in {\cal H}^{(F)}
 \,.
 \label{eqF}
 \ee
Here, one works in another, auxiliary Hilbert space which is, by
construction, unphysical. Its symbol is characterized by
superscripted $^{(F)}$ which abbreviates ``friendly'' as well as
``false'' \cite{SIGMA}.

Although it gave the name to the theory, it is not too relevant that
the transformed Hamiltonian is manifestly non-Hermitian, i.e., that
we have $H_{(SRNH)} \neq H_{(SRNH)}^\dagger$ in the
$F-$superscripted Hilbert space. The reason is that the latter
Hilbert space is purely auxiliary and that it is not intended to
carry any quantum theoretical meaning by itself. What is only
important is that in the light of the underlying theory the
necessary physical meaning may easily be reinstalled if one
redefines the inner product \cite{Carl,ali}.  In mathematical
language this means that one has to replace the unphysical
$F-$superscripted Hilbert space by its ``standard'' physical
amendment ${\cal H}^{(S)}$. The latter space differs from its
predecessor ${\cal H}^{(F)}$ strictly and solely by its upgraded,
metric-dependent inner product between any pair of its {\em shared}
ket-vector elements,
 \be
 \br {\psi}_1
 |\psi_2\kt_{(S)} \ \equiv \
  \br \psi_1 |\Theta
 |\psi_2\kt_{(F)} \,,
  \ \ \ \ |\psi_{1,2}\kt \in {\cal H}^{(F,S)}\,.
 \label{rewpro}
 \ee
The mathematical core of the whole SRNH scheme of
Refs.~\cite{Carl,ali,Geyer} just follows from the obvious identity
 \be
 \br \varphi_1| \varphi_2 \kt =
 \br \psi_1 |\Theta
 |\psi_2\kt = \br \tilde{\psi}_1
 |\psi_2\kt\,,
 \ \ \ \ | \varphi_{1,2} \kt  \in {\cal H}^{(P)}
 \,,
 \ \ \ \ |\psi_{1,2}\kt,
 \, |\tilde{\psi}_{1}\kt \in {\cal H}^{(F)}\,.
 \label{newpro}
 \ee
These relations may be read as a proof of the strict unitary
equivalence ${\cal H}^{(S)}\ \equiv \ {\cal H}^{(P)}$ between the
two alternative physical Hilbert spaces. As an immediate consequence
one deduces the indistinguishability between the $P-$based and
$S-$based physical probabilistic predictions.

One of the important advantages of the SRNH formalism is that in the
light of relations (\ref{newpro}) one can completely suppress the
reference to our ultimate physical $S-$superscripted SRNH Hilbert
space once one decides to use the nontrivial metric, whenever
needed, in explicit manner. Such a convention is also being accepted
in the majority of the newer SRNH applications as reviewed in
Refs.~\cite{Carl,ali} or \cite{MZbook}.

In the more pragmatic and phenomenological physical setting it makes
sense to recall, once more, the instructive nuclear-physics
illustration of Ref.~\cite{Geyer} in which the initial space ${\cal
H}^{(P)}$ had the form of a Pauli-antisymmetrized Fock space of the
half-spin fermions while the (shared) ket-vector elements of
Dyson-partner spaces ${\cal H}^{(F,S)}$ were the states which obeyed
the perceivably simpler statistics of effective bosons. Another
important class of implementations of the whole SRNH pattern are the
systems in which the non-Hermitian Hamiltonians in Schr\"{o}dinger
Eq.~(\ref{eqF}) are pseudo-Hermitian {\it alias} ${\cal PT}$
symmetric (cf. reviews \cite{Carl,ali,MZbook} for more details).

\section{Time-dependent maps $\Omega(t)$}

In the conventional, Hermitian quantum models the operators of
observables $\mathfrak{a}$ are usually assumed, in the SREP setting,
time-independent, $\mathfrak{a}^{(SREP)}(t)=
\mathfrak{a}^{(SREP)}(0)$. Less often, the time-dependent influence,
say, of a variable external field is reflected by an induced
time-dependence in $\mathfrak{a}^{(SREP)}(t)$. In such a case, it
must be treated as an independent dynamical input information.
Hence, it cannot follow from any equations and must be prescribed in
full extent.

In contrast, what is assumed time-independent in the conventional
Hermitian HREP picture are wave functions,
 \be
 |\varphi^{(HREP)}(t)\kt = |\varphi^{(HREP)}(0)\kt\,.
 \label{rovnajse}
 \ee
The necessary unitary transformation which mediates the transition
between the SREP and HREP pictures is elementary. Its form may be
found in any textbook \cite{Messiah}.

Our present study is inspired by the fact that the conventional
unitary SREP $\to $ HREP transition may be given the generic form of
Eq.~(\ref{2}). Unfortunately, this seems to be in conflict with the
current literature on non-Hermitian models. Indeed, up to a few
inconclusive exceptions \cite{Bila}  there emerged virtually no
studies using time-dependent non-unitary maps $\Omega(t)$ and/or
nontrivial metrics $\Theta(t)$. In \cite{ali}, moreover, it has been
shown that the stationarity of metrics is, up to an overall phase,
unavoidable whenever one requires that the time-evolution generator
$H$ entering SRNH Eq.~(\ref{eqF}) represents an observable quantity.

In a way explained in Ref.~\cite{timedep} we know that the
requirement of the observability of the time-evolution generator is
rather artificial and that its removal leaves the whole SRNH theory
fully mathematically consistent. Simultaneously, we are fully aware
of the fact that many practical merits of the SRNH recipes (and,
first of all, the possibility of a simplification of
Eq.(\ref{uneSEtaune})) get definitely lost for nontrivial $\Omega(t)
\neq \Omega(0)$ \cite{Coriolis}. In the context of such a conflict
of our own opinions the core of our present message is twofold.
Firstly, we shall outline the NHH formalism in a concrete form and,
secondly, we shall demonstrate that its technical aspects are,
especially in comparison with the conventional HREP, quite
user-friendly and not untractable.

\subsection{General case}

In review~\cite{SIGMA} we started from the same Schr\"{o}dinger
Eq.~(\ref{uneSEtaune}) as above. The change of the perspective was
that the $F \to P$ mapping was admitted arbitrarily time-dependent,
$\Omega=\Omega(t)$. This enabled us to upgrade definition (\ref{2}),
 \be
 |\varphi^{(SREP)}(t)\kt=
 \Omega(t)|\psi(t)\kt=
 \Omega^\dagger(t)|\tilde{\psi}(t)\kt  \in {\cal
 H}^{(P)}\,,
 \ \ \ \ \ \
 |{\psi}(t)\kt,
 |\tilde{\psi}(t)\kt \in {\cal H}^{(F)}\,.
 \ee
Next, still without any direct reference to the prospective NHH
context we set
 \be
 \mathfrak{h}_{(SREP)}=\mathfrak{h}_{(SREP)}(t)
   = \Omega(t)\,H(t) \,\Omega^{-1}(t)\,
   \label{haham}
 \ee
and inserted all of these ansatzs in Eq.~(\ref{uneSEtaune}).
Elementary differentiation led to the doublet of $F-$space
Schr\"{o}dinger equations
 \be
 {\rm i}\p_t|\psi(t)\kt = G(t)\,|\psi(t)\kt\,,
 \ \ \ \ \ \
 |\psi(t)\kt \in {\cal H}^{(F)}\,,
 \label{SEtaujoa}
 \ee
 \be
 {\rm i}\p_t|\tilde{\psi}(t)\kt = G^\dagger(t)\,|\tilde{\psi}(t)\kt\,,
 \ \ \ \ \ \
 |\tilde{\psi}(t)\kt \in {\cal H}^{(F)}\,.
 \label{SEtaujobe}
 \ee
We emphasized that the generator $ G(t) = H(t) - \Sigma(t)$ of the
time-evolution of the wave functions differed from the non-Hermitian
energy-operator of Eq.~(\ref{haham}) by the easily evaluated
additional Coriolis-force term of a purely kinematical origin
\cite{Coriolis},
 \be
 \Sigma(t)
 = {\rm i}\Omega^{-1}(t)\left [\partial_t\Omega(t)\right ]
 = {\rm i}\Omega^{-1}(t)\dot{\Omega}(t)\,.
 \label{diffe}
  \ee
Now we may add that the emergence of the Coriolis force changes the
time-dependence of any operator $\mathfrak{a}$ of an observable. Its
$P-$space representation $\mathfrak{a}(t)= \mathfrak{a}^\dagger(t)$
as well as its $F-$space image
$A(t)=\Omega^{-1}(t)\mathfrak{a}(t){\Omega}(t)$ may be
differentiated yielding formula
 \be
 {\rm i}\partial_t A(t) = A(t) \Sigma(t)-\Sigma(t) A(t)
 +\Omega^{-1}(t)[{\rm i}\dot{\mathfrak{a}}(t)]{\Omega}(t)\,.
 \label{devet}
 \ee
This formula will be a starting point of our forthcoming
considerations. It is important to notice that in its general form
it is still too formal since it intermingles the information encoded
in the $P-$space operator $\mathfrak{a}$ with the one  encoded in
its $P-$space transform $A$ and with the one encoded in the Dyson's
operator $\Omega$ itself.

\subsection{Heisenberg representations}

The characteristic HREP postulate (\ref{rovnajse}) that the wave
functions of a state remain constant in time means, in the general
three-Hilbert-space context of Eqs.~(\ref{SEtaujoa}) and
(\ref{SEtaujobe}), that in the prospective NHH setting we have to
select, first of all, the trivial generator, $G^{(NHH)}(t) \equiv
0$. Such a constraint implies the necessity of the coincidence
between the Coriolis NHH force and the NHH observable-energy
operator,
 \be
 \Sigma_{(NHH)}(t)=H_{(NHH)}(t)\,.
 \label{charhr}
 \ee
Thus, given the energy operator $H_{(NHH)}(t)$ as a dynamical input
information we may reconstruct the specific Dyson's NHH operator
$\Omega(t)=\Omega_{(NHH)}(t)$ via Eq.~(\ref{diffe}). The NHH version
of this operator differential equation reads
 \be
  {\rm i}\partial_t\Omega_{(NHH)}(t)
 =\Omega_{(NHH)}(t)\,H_{(NHH)}(t)\,.
 \label{differe}
  \ee
Thus, a time-dependent NHH Dyson map is defined, in principle at
least. As an operator Cauchy problem the construction must be
initiated by a suitable preselected non-unitary operator
$\Omega^{(NHH)}(t_0)$. The latter initial value must be such that
the NHH Hamiltonian observability remains guaranteed,
 \be
 H_{(NHH)}^\dagger(t_0)\ \Theta^{(NHH)}(t_0)
 \ =\ \Theta^{(NHH)}(t_0)\,H_{(NHH)}(t_0)\,.
 \label{Diehor}
 \ee
We may conclude that the main difference between the Hermitian and
non-Hermitian HREP constructions is that the operator solution
$\Omega^{(NHH)}(t)$ of Eq.~(\ref{differe}) is not unitary. Thus, in
the NHH setting it is now time to add another simplifying
assumption:

\begin{defn}
 \label{defone}
The special stationary subfamily of the general NHH Hamiltonians of
Eq.~(\ref{charhr}) will be specified by requirement $H_{(NHH)}(t) =
H_{(NHH)}(0)$ and marked by a dedicated subscript,
 $
 H_{(NHH)}(0)=H_{(spec)}$.
\end{defn}
It is worth noticing that even in the conventional Hermitian HREP
recipe the use of stationary Hamiltonians simplifies a number of
technicalities. The same holds true in the NHH setting. First of
all, it is easy to prove the following

\begin{lemma}

In the stationary special case of Definition \ref{defone} our
operator differential Eq.~(\ref{differe}) may be solved in closed
form,
 \be
 \Omega_{(spec)}(t)=\Omega_{(spec)}(t_0)\,e^{-{\rm i}\,(t-t_0)\,
 H_{(spec)}}\,.
 \label{ome}
 \ee
The initial value must be compatible with the factorization of
metric in Eq.~(\ref{Diehor}),
 \be
 \Theta^{(spec)}(t_0)=\Omega_{(spec)}^\dagger(t_0)\Omega_{(spec)}(t_0)\,.
 \label{factspec}
 \ee
 \label{propone}
 \end{lemma}
The following result is slightly less expectable.

\begin{prop}
Although the NHH Dyson's map of Eq.~(\ref{ome}) varies with time it
still leads to the metric which is time-independent,
$\Theta^{(spec)}(t)=\Theta^{(spec)}(t_0)$.
 \label{proptwo}
 \end{prop}
 \begin{proof}
In the definition
 \be
 \Theta^{(spec)}(t)=e^{{\rm i}\,(t-t_0)\, H_{(spec)}^\dagger\,}\,
 \Theta^{(spec)}(t_0)\,e^{-{\rm i}\,(t-t_0)\, H_{(spec)}}
 \,
 \label{Dieh}
 \ee
we expand one of the exponentials in the operator Taylor series.
Next, we recall the Hamiltonian-observability initial condition
 \ben
 H_{(spec)}^\dagger\,\Theta^{(spec)}(t_0)
 =\Theta^{(spec)}(t_0)\,H_{(spec)}\,
 \een
and intertwine all of the powers of the time-independent
$H_{(spec)}$ with the central operator of the initial-value metric.
Finally, we re-convert the modified Taylor series into the operator
exponential and multiply the two time-dependent factors yielding the
unit operator.
 \end{proof}

We see that our stationary operator of metric defines the physical
inner product in Hilbert space $ {\cal H}^{(S)}={\cal H}^{(spec)}$.
In the light of the construction the product and, hence, also the
space remain time-independent. This has further nontrivial
consequences in ${\cal H}^{(P)}$.

\begin{prop}
In the stationary NHH scenario also the self-adjoint isospectral
partner Hamiltonian operator remains stationary.

\end{prop}

\begin{proof}
In order to prove the time-independence of
 \be
 \mathfrak{h}_{(spec)}(t)=\Omega_{(spec)}(t)\,H_{(spec)}
 \left [\Omega_{(spec)}(t)
 \right ]^{-1}
 = \mathfrak{h}_{(spec)}^\dagger(t)\,,
 \ee
it is sufficient to proceed in analogy with the proof of Prop.
\ref{proptwo}. The situation is now simpler because the
time-dependent exponentials commute with the Hamiltonian. Hence,
they immediately cancel and no operator-intertwining is needed.
\end{proof}

\subsection{Heisenberg equations of motion}

In the special NHH stationary case of Definition \ref{defone} the
knowledge of the mapping $\Omega_{(spec)}(t)$ of Eq.~(\ref{ome})
together with the knowledge of  an operator $\mathfrak{a}(t)$ of any
observable in ${\cal H}^{(P)}$ enables us to reconstruct the action
$A(t)$ of the same observable in the auxiliary, unphysical space $
{\cal H}^{(F)}$. The latter space merely differs from the correct
and physical Hilbert space $ {\cal H}^{(S)}={\cal H}^{(spec)}$ by
the ``false'' metric $\Theta^{(F)}=I$ so that the same formula
 \be
 A^{(spec)}(t)=e^{{\rm i}\,(t-t_0)\, H_{(spec)}}\,
 \Omega_{(spec)}^{-1}(t_0)\,
 \mathfrak{a}(t)\,
 \Omega_{(spec)}(t_0)\,
 e^{-{\rm i}\,(t-t_0)\, H_{(spec)}}
 \,
 \label{Diehphu}
 \ee
offers the explicit definition. No equation of motion is needed.
Naturally, whenever asked for, we may still recall the universal
rule (\ref{devet}) and write down, after appropriate substitutions,
the equation of motion
 \be
 {\rm i}\partial_t A_{(spec)}(t) = A_{(spec)}(t)  H_{(spec)}-
 H_{(spec)} A_{(spec)}(t)
 +K(t)\,.
 \label{joevet}
 \ee
Here, an additional and independent dynamical input must be provided
via operator
 \ben
 K(t)=\Omega^{-1}_{(spec)}(t)[{\rm
 i}\dot{\mathfrak{a}}(t)]{\Omega}_{(spec)}(t)\,.
 \een
Unless we select dynamical scenario in which $\partial_t
{\mathfrak{a}}(t)=0$, recipe (\ref{joevet}) remains highly formal.
The knowledge of the time-derivative of observable $\mathfrak{a}$ in
${\cal H}^{(P)}$ (and of its initial value at any $t_{ini}$) could
have been much more easily complemented by an exhaustive
reconstruction of operator $\mathfrak{a}(t)$ and by its insertion,
at all times, in definition (\ref{Diehphu}) of $A^{(spec)}(t)$.

In the sense of this comment the more general non-stationary cases
which do not obey Definition \ref{defone} but which still do not
need any additional information are much more interesting.

\begin{defn}

The observables exhibiting the manifest-time-independence property
(i.e., satisfying relation $\partial_t
{\mathfrak{a}}_{(indep)}(t)=0$ in ${\cal H}^{(P)}$) will be marked
by subscript $_{(indep)}$.
 \label{deftwo}
\end{defn}

\begin{thm}

Under the constraint imposed by Definition \ref{deftwo}, the
observable $A_{(indep)}(t)$ which is defined as acting in the
physical NHH Hilbert space ${\cal H}^{(S)}={\cal H}_{(NHH)}$ may be
determined from the pair of operator differential evolution
equations
 \be
 {\rm i}\partial_t A_{(indep)}(t) = A_{(indep)}(t)
 H_{(NHH)}(t)-H_{(NHH)}(t) A_{(indep)}(t)\,
 \label{devetak}
 \ee
and
 \be
 {\rm i}\partial_t A^\dagger_{(indep)}(t) =
 A^\dagger_{(indep)}(t)
 H^\dagger_{(NHH)}(t)-H^\dagger_{(NHH)}(t) A^\dagger_{(indep)}(t)\,
 \label{devetakonj}
 \ee
together with the obligatory coupling
 \be
 A^\dagger_{(indep)}(t_0)\Theta_{(NHH)}(t_0)
 =\Theta_{(NHH)}(t_0)A_{(indep)}(t_0)\,
 \label{suff}
 \ee
between the respective initial operator values.

\end{thm}

\begin{proof}
The assumption of the $P-$space stationarity
$\partial_t{\mathfrak{a}}_{(indep)}(t)=0$ enables us to omit the
last term from Eq.~(\ref{devet}) living in friendly space ${\cal
H}^{(F)}$. In combination with Eq.~(\ref{charhr}) this leads to the
desired pair of generalized non-Hermitian Heisenberg equations. The
Hermitian-conjugation duplicity is nontrivial because we must
guarantee the observability of $A_{(indep)}(t)$ at all times, i.e.,
the validity of relation
 \be
 A^\dagger_{(indep)}(t)\Theta_{(NHH)}(t)
 =\Theta_{(NHH)}(t)A_{(indep)}(t)\,.
 \label{suffic}
 \ee
This relation reflects the $P\to F$-inherited Hermiticity {\it
alias} crypto-Hermiticity of our observable. Once we differentiate
this relation with respect to time we get a new relation which is
identically satisfied. This is immediately proved when we recall (or
quickly derive) an insert the well known \cite{Bila} elementary
identity
 \be
 {\rm i}\partial_t \Theta(t) = \Theta(t) \Sigma(t) -
 \Sigma^\dagger(t)  \Theta(t)
 \,.
 \label{deonj}
 \ee
Thus, it is necessary and sufficient to postulate Eq.~(\ref{suffic})
at $t=t_0$.
\end{proof}

\section{Summary}

The use of non-unitary operators $\Omega$ in the time-independent
and time-dependent regimes is not too dissimilar. One simply works
with the $P-F-S$ triplet of Hilbert spaces in a way which is
summarized by the following diagram,
 \ben
  %\vspace{-1cm}
  \ba
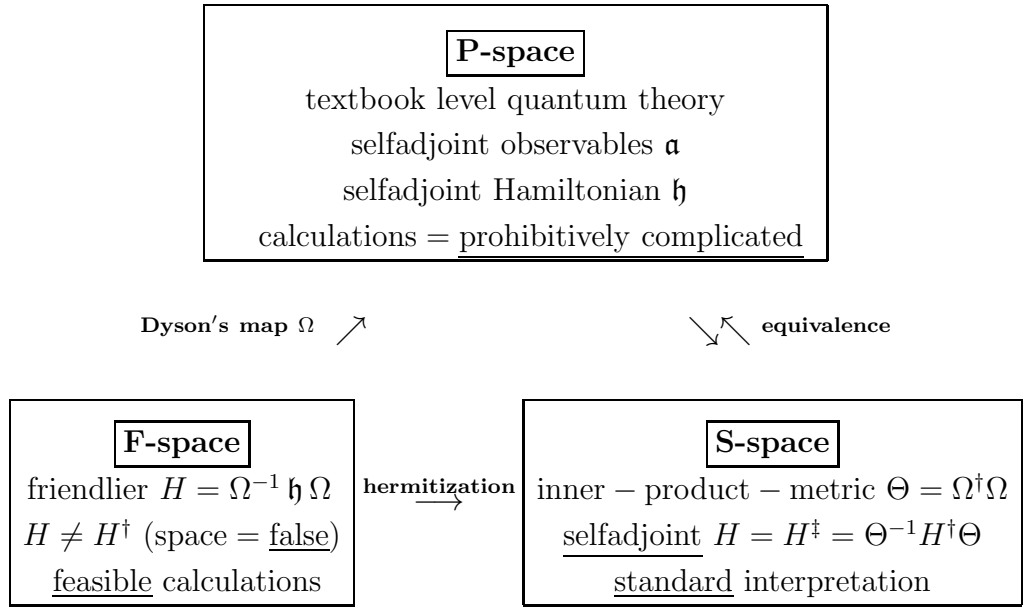
    \begin{array}{|c|}
 \hline
 \vspace{-0.3cm}\\
  {\rm  \fbox{\bf
  %\textcolor{red}
  {P}-space}}\\
  \ {\rm textbook \  level\ quantum\ theory}\ \\
 {\rm selfadjoint\   observables}\ \mathfrak{a}\\
 %{\rm  principle-of-correspondence}\ \\
 {\rm selfadjoint\  Hamiltonian}\ \mathfrak{h}\\
 %{\rm generates \ unitary\  time\ evolution}\ \\
  \ \ \ \ {\rm  calculations= \underline{prohibitively\ complicated}  } \ \\
    % {\rm  \ quantum\ {
%     \underline{physics}}\ of \ traditional \ textbooks } \ \\
 %\ \ \
 %  {\rm   physical\ meaning \ obvious}\  \
 \hline
 \ea
 \\
 \\
 \stackrel{{\bf  Dyson's \ map}\ \Omega}{}
 \ \
  \nearrow\ \  \  \ \ \ \ \ \
 \ \ \ \ \ \ \ \
  \ \  \  \ \ \ \ \ \
 \ \ \ \ \  \searrow \nwarrow\
 \stackrel{\bf   equivalence}{}\\
 \\
 \begin{array}{|c|}
 \hline
 \vspace{-0.35cm}\\
  {\rm  \fbox{\bf
  %\textcolor{red}
  {F}-space}}\\
   {\rm  friendlier}\  H=\Omega^{-1}\,\mathfrak{h}\,\Omega\,  \\
  H \neq H^\dagger \ {\rm (space=\underline{false})} \\
    \ {\rm {  \underline{feasible}\ calculations}}  \\
  \hline
 \ea
 \stackrel{ {\bf  hermitization}  }{ \longrightarrow }
 \begin{array}{|c|}
 \hline
 \vspace{-0.35cm}\\
  {\rm  \fbox{\bf
  %\textcolor{red}
  {S}-space}}\\
      {\rm  inner-product-metric\ } \Theta=\Omega^\dagger\Omega \\
   \underline{\rm  selfadjoint}\   H=H^\ddagger=\Theta^{-1}H^\dagger\Theta  \\
  {\rm  \underline{standard}\ interpretation}
 \\
 \hline
 \ea
\\
\\
\ea
 \een
In this framework we described a few basic features of the
innovated, non-Hermitian Heisenberg-representation approach to
quantum theory. In a restricted format of the letter we managed to
demonstrate not only its mathematical consistency but also its
sufficiently friendly and feasible nature.

\end{document}